\documentclass[a4paper,USenglish,cleveref, autoref, thm-restate]{lipics-v2021}

\usepackage{mathtools}
\usepackage{nicefrac}
\usepackage[ruled,vlined,linesnumbered]{algorithm2e}
\usepackage{csquotes}

\usepackage{tikz}
\usetikzlibrary{decorations.pathreplacing}
\usetikzlibrary{math}
\usepackage{xspace}

\usepackage{cite}

\newcommand{\problem}{\textnormal{USM}\xspace}
\newcommand{\eps}{\varepsilon}

\newcommand{\Opt}{\operatorname{\text{\textsc{opt}}}}

\newcommand{\machs}{\mathcal{M}}
\newcommand{\jobs}{\mathcal{J}}
\newcommand{\maxload}{\ell_{\max}}
\newcommand{\load}{\ell}
\newcommand{\oldload}{\hat{\ell}}
\newcommand{\newload}{\check{\ell}}
\newcommand{\startload}{\bar{\ell}}
\newcommand{\shiftload}{\tilde{\ell}}
\newcommand{\oldjobs}{\hat{\jobs}}
\newcommand{\newjobs}{\check{\jobs}}
\newcommand{\potential}{\pi}
\newcommand{\Oh}{\mathcal{O}}

\DeclarePairedDelimiter\set{\lbrace}{\rbrace}
\DeclarePairedDelimiterX\sett[2]{\lbrace}{\rbrace}{ #1 \,\delimsize| \,\mathopen{} #2 }

\hideLIPIcs  

\title{Online Load Balancing on Uniform Machines with Limited Migration} 

\author{Marten Maack}{Heinz Nixdorf Institute \& Department of Computer Science, Paderborn University, Paderborn, Germany}{marten.maack@hni.uni-paderborn.de}{https://orcid.org/0000-0001-7918-6642}{}

\authorrunning{M. Maack} 

\Copyright{Marten Maack} 



\category{} 

\relatedversion{} 

\funding{This work was supported by the German Research Foundation (DFG) within the project "Robuste Online-Algorithmen für Scheduling- und Packungsprobleme" under the project number JA 612 /19-1 and within the Collaborative Research Centre “On-The-Fly Computing” under the project number 160364472 -- SFB 901/3.}

\acknowledgements{}

\nolinenumbers 

\EventEditors{John Q. Open and Joan R. Access}
\EventNoEds{2}
\EventLongTitle{42nd Conference on Very Important Topics (CVIT 2016)}
\EventShortTitle{CVIT 2016}
\EventAcronym{CVIT}
\EventYear{2016}
\EventDate{December 24--27, 2016}
\EventLocation{Little Whinging, United Kingdom}
\EventLogo{}
\SeriesVolume{42}
\ArticleNo{23}

\begin{document}

\maketitle

\begin{abstract}
In the problem of online load balancing on uniformly related machines with bounded migration, jobs arrive online one after another and have to be immediately placed on one of a given set of machines without knowledge about jobs that may arrive later on.
Each job has a size and each machine has a speed, and the load due to a job assigned to a machine is obtained by dividing the first value by the second.
The goal is to minimize the maximum overall load any machine receives.
However, unlike in the pure online case, each time a new job arrives it contributes a migration potential equal to the product of its size and a certain migration factor. 
This potential can be spend to reassign jobs either right away (non-amortized case) or at any later time (amortized case).
Semi-online models of this flavor have been studied intensively for several fundamental problems, e.g., load balancing on identical machines and bin packing, but uniformly related machines have not been considered up to now.
In the present paper, the classical doubling strategy on uniformly related machines is combined with migration to achieve an $(8/3+\eps)$-competitive algorithm and a $(4+\eps)$-competitive algorithm with $\Oh(1/\eps)$ amortized and non-amortized migration, respectively, while the best known competitive ratio in the pure online setting is roughly $5.828$.
\end{abstract}

\section{Introduction}
\label{sec:intro}

Consider the problem of online load balancing on uniformly related machines with bounded migration (or \problem for short):
Jobs from a set $\jobs$ arrive one after another and have to be immediately placed on a machine from a set $\machs$ without knowledge about jobs that may arrive later on.
Each job $j$ has a size $p_j$, each machine $i$ a speed $s_i$, and the load of job $j$ on machine $i$ is given by $p_j/s_i$.
The load $\ell_i$ of a machine $i$ is the summed up load of all the jobs assigned to $i$ and the goal is to minimize the maximum machine load $\load_{\max} = \max_{i\in\machs}\load_i$. 
This would conclude the description of the problem in the pure online setting, but in this work a semi-online model is considered.
In particular, the arrival of a new job $j$ allows the reassignment of jobs with total size at most $\mu p_{j}$, i.e., we have a \emph{migration potential} of $\mu p_{j}$ due to $j$.
The parameter $\mu$ is called the \emph{migration factor} and we consider two variants, i.e., either the full migration potential of a job $j$ has to be used directly after the arrival of $j$ or any time later.
We refer to the former case as non-amortized and to the second as amortized and talk about the (non-)amortized migration factor as well.
Like in the pure online setting, we consider the competitiveness of an algorithm as a measure of quality, i.e., the worst case ratio between the objective value of a solution found by the algorithm and the optimal objective value for a given instance.
Additionally, we want to bound the migration factor and therefore consider a trade-off between this factor and the competitive ratio.

\subparagraph{Literature Review.} 
Regarding the pure online setting (without migration), there is a classical result due to Aspnes et al.\xspace\cite{DBLP:conf/stoc/AspnesAFPW93} that utilizes a guess-and-double framework to achieve an 8-competitive algorithm for the problem.
The best competitive ratio of $\approx 5.828$ was achieved by Bar-Noy et al.\xspace\cite{BarNoyFN00}, while Ebenlendr and Sgall \cite{DBLP:journals/mst/EbenlendrS15} showed that no better competitive ratio than $\approx 2.564$ can be achieved.
Regarding non-amortized migration, Sanders et al.\xspace\cite{DBLP:journals/mor/SandersSS09} considered the identical machine case where all machine speeds are equal.
In particular, they designed several simple algorithms with small competitive ratios and small non-amortized migration factors as well as a family of $(1+\eps)$-competitive algorithms with a non-amortized migration factor exponential in $1/\eps$.
Skutella and Verschae \cite{DBLP:journals/mor/SkutellaV16} generalized the later result to the dynamic case where jobs may also depart but using amortized migration.
Uniformly related machines in the dynamic case and with respect to amortized migration have been studied by Westbrook \cite{DBLP:journals/jal/Westbrook00} who provided an $(8+\eps)$-competitive algorithms with an amortized migration factor polynomial in $1/\eps$ for each $\eps >0$.
Finally, Andrews et al. \cite{DBLP:journals/algorithmica/AndrewsGZ99} gave a $32$-competitive algorithm using a generalized notion of amortized migration -- where the migration cost may not be proportional to job size -- and achieving a factor of $79.4$.
There are many more results considering the (non-)amortized migration factor primarily regarding scheduling and packing problems (see, e.g., \cite{DBLP:conf/waoa/BerndtDGJK19,DBLP:journals/mp/BerndtJK20,DBLP:journals/talg/GalvezSV20,DBLP:conf/icalp/FeldkordFGGKRW18,DBLP:journals/orl/BerndtEM22,DBLP:conf/soda/GuptaKS14,DBLP:journals/mp/EpsteinL09,DBLP:journals/siamjo/EpsteinL13,DBLP:conf/esa/EpsteinL11}).

\subparagraph{Results.}

In this work, we revisit the online doubling approaches \cite{DBLP:conf/stoc/AspnesAFPW93,BarNoyFN00} utilizing migration.
In particular, we develop an algorithmic framework combining the guess-and-double approach with a straightforward migration strategy. 
There are slight differences depending on whether amortized or non-amortized migration is used.
Otherwise, the framework is governed by three parameters influencing the doubling rate, amount of migration, and precise placement strategy.
Using the arguably most intuitive choice for one of these parameters leads to the first result, namely, a $(3+\eps)$-competitive algorithm for each $\eps>0$ with amortized migration factor $\Oh(1/\eps)$.
The attempt of transferring this result to the case with non-amortized migration motivates the closer examination of the said parameter, which enables the second result -- a $(4+\eps)$-competitive algorithm for each $\eps \in (0,8]$ with non-amortized migration factor $\Oh(1/\eps)$.
Finally, taking insights from the second result back to the first yields an $(8/3+\eps)$-competitive algorithm for each $\eps >0$ with amortized migration factor $\Oh(1/\eps)$. 
Hence, all three results clearly beat the best pure online approach with very moderate migration with the last result approaching the lower bound for the pure online case.

\section{Algorithmic Framework}
\label{sec:algo}

Throughout this paper, we assume $\machs=[m]$, $s_1\geq s_2\geq \dots\geq s_m$, and denote the optimal (offline) objective value of the input instance $I$ as $\Opt(I)$.
Furthermore, we typically consider a new job $j^*$ that has to be scheduled and denote the instance and the job set up to and including $j^*$ as $I^*$ and $\jobs^*$, respectively.

We start this section with a short account of the classical online algorithm by Aspnes et al.\xspace\cite{DBLP:conf/stoc/AspnesAFPW93},
then introduce the central ideas and notation leading to the description of the algorithmic framework.

\subparagraph{The Classical Online Algorithm.}

The following approach is due to Aspnes et al.\xspace\cite{DBLP:conf/stoc/AspnesAFPW93}.
If the optimum objective value $\Opt(I)$ of the instance $I$ is known, a schedule with maximum load at most $2\Opt(I)$ can be constructed online by assigning the jobs to the slowest machine on which it may be processed without exceeding the load bound of $2\Opt(I)$.
This can be used to achieve an $8$-competitive algorithm via a guess-and-double strategy that works in phases.
In each phase the above algorithm is applied using a guess $T$ of $\Opt(I)$ and not taking into account the load that was accumulated in the previous phases.
When the first job $j_1$ arrives we set $T=p_1/s_1$ and apply the above algorithm until a job $j^*$ cannot be inserted on any machine without exceeding the bound $2T$.
This gives a proof that $T < \Opt(I^*)$.
We then double the guess $T$ and start the next phase.
During each phase the load on any machine can obviously be upper bounded by $2T\sum_{i=0}^{\infty}1/2^i=4T$.
So, after scheduling the current job $j^*$, we have $\maxload < 4T$ and $1/2T<\Opt(I^*)$, yielding $\maxload< 8\Opt(I^*)$.
This is depicted in the first part of \cref{fig:load_structure}.
\begin{figure}
\centering
\begin{tikzpicture}[scale=0.4]
\tikzmath{
\xx3 = 4.5; 
\x9 = 0;
\xx1 = \x9 + \xx3;
\xx2 = \xx1 / 2;
\x1 = \xx1 + 5.5; 
\x2 = \x1 + \xx3;
\x3 = \x1 - 2;
\x4 = \x2 + 1.5;
\x5 = \x2 + 6;
\x6 = \x5 + \xx3;
\x7 = \x5 - 2;
\x8 = \x6 + 1.5;
}
\pgfmathsetmacro{\epsi}{0.18}


\draw[dashed,thick] (\x9,0) rectangle (\xx1  ,1);

\draw[fill = black!7!white,rounded corners] (\x9 + \epsi,1 + \epsi) rectangle (\xx1 -\epsi , 2 - \epsi);
\draw[fill = black!7!white,rounded corners] (\x9 + \epsi,2 + \epsi) rectangle (\xx1 -\epsi , 3.5 - \epsi);
\draw[fill = black!7!white,rounded corners] (\x9 + \epsi,4 + \epsi) rectangle (\xx1 -\epsi , 7.0 - \epsi);
\draw[fill = black!7!white,rounded corners] (\x9 + \epsi,8 + \epsi) rectangle (\xx1 -\epsi , 14.5 - \epsi);

\foreach \i in {0,...,3}
{	
\draw[thick] (\x9,{2^\i}) rectangle (\xx1  , {2^(\i+1)});
\draw[thick, dashed] (\x9,{3*2^(\i-1)}) -- (\xx1 ,{3*2^(\i-1)});
}
\foreach \i in {1,2,3}
{
	\draw [fill] (\xx2,1 - \i * 0.2) circle [radius=0.02];
}
\draw[thick,decorate,decoration={brace,amplitude=10pt,raise=1pt,mirror},yshift=0pt] (\xx1 ,0) -- (\xx1 ,16) node [midway,xshift=20pt]{$4T$};
\draw[thick,decorate,decoration={brace,amplitude=6pt,raise=1pt},yshift=0pt] (\x9,8) -- (\x9,12) node [midway,xshift=-15pt]{$T$};


\draw[fill = black!20!white,rounded corners] (\x1 + \epsi,\epsi) rectangle (\x2 - \epsi, 2 - 2*\epsi) node [midway]{$\wedge$};
\draw[fill = black!20!white,rounded corners] (\x1+ \epsi,8- \epsi) rectangle (\x2- \epsi, 5+ \epsi) node [midway]{$\vee$};
\draw[fill = black!20!white,rounded corners] (\x1+ \epsi,5) rectangle (\x2- \epsi, 2+2*\epsi) node [midway, yshift = - 0.1 cm]{$\vee$};
\draw[thick] (\x1,0) rectangle (\x2, 8);
\draw[thick, dashed] (\x1,2) -- (\x2,2);
\draw[thick, dashed] (\x1,4) -- (\x2,4);
\draw[thick, ->, >=latex] (\x3,0) -> (\x3,4) node [midway,xshift=-10pt]{$\oldjobs_i$};
\draw[thick, ->, >=latex] (\x4,8) -> (\x4,0) node [midway,xshift=10pt]{$\newjobs_i$};
\draw[thick,decorate,decoration={brace,amplitude=3pt,raise=1pt},yshift=0pt] (\x1,0) -- (\x1,2) node [midway,xshift=-12pt]{$\frac{1}{2}T$};
\draw[thick,decorate,decoration={brace,amplitude=4pt,raise=1pt,mirror},yshift=0pt] (\x2,4) -- (\x2,8) node [midway,xshift=10pt]{$T$};


\draw[fill = black!7!white,rounded corners] (\x5 + \epsi,\epsi) rectangle (\x6 - \epsi, 1.5) node [midway, yshift = - 0.1 cm]{$\wedge$};
\draw[fill = black!7!white,rounded corners] (\x5 + \epsi,12- \epsi) rectangle (\x6 - \epsi, 6+ \epsi) node [midway]{$\vee$};
\draw[thick] (\x5,0) rectangle (\x6, 12);
\draw[thick, dashed] (\x5,1) -- (\x6,1);
\draw[thick, dashed] (\x5,2) -- (\x6,2);
\draw[thick, dashed] (\x5,3) -- (\x6,3);
\draw[thick, dashed] (\x5,4) -- (\x6,4);
\draw[thick, dashed] (\x5,8) -- (\x6,8);
\draw[thick, ->, >=latex] (\x7,0) -> (\x7,3) node [midway,xshift=-10pt]{$\oldjobs_i$};
\draw[thick, ->, >=latex] (\x8,12) -> (\x8,0) node [midway,xshift=10pt]{$\newjobs_i$};
\draw[thick,decorate,decoration={brace,amplitude=2pt,raise=1pt},yshift=0pt] (\x5,0) -- (\x5,1) node [midway,xshift=-12pt]{$\frac{1}{4}T$};
\draw[thick,decorate,decoration={brace,amplitude=4pt,raise=1pt},yshift=0pt] (\x6,12) -- (\x6,8) node [midway,xshift=10pt]{$T$};

\end{tikzpicture}
\caption{Sketches of the load on a single machine $i$ for the classical online algorithm, the first amortized approach and the non-amortized approach depicted from left to right.
Regarding the classical online algorithm the light gray boxes representing sets of jobs.
For the first amortized approach, we have $\xi =2$, $\gamma = 2/3$, and the critical case with two new jobs of size roughly $0.75T$ and one remaining old job of size roughly $0.5T$ is depicted (with the gray boxes representing the jobs).
For the non-amortized approach, we have $\xi = 4$, $\eta =2$, $\gamma = 1/2$, and a non-critical case is depicted where half of the possible new load arrived and half of the old load was migrated (with the light gray boxes representing sets of jobs).
}
\label{fig:load_structure}
\end{figure}
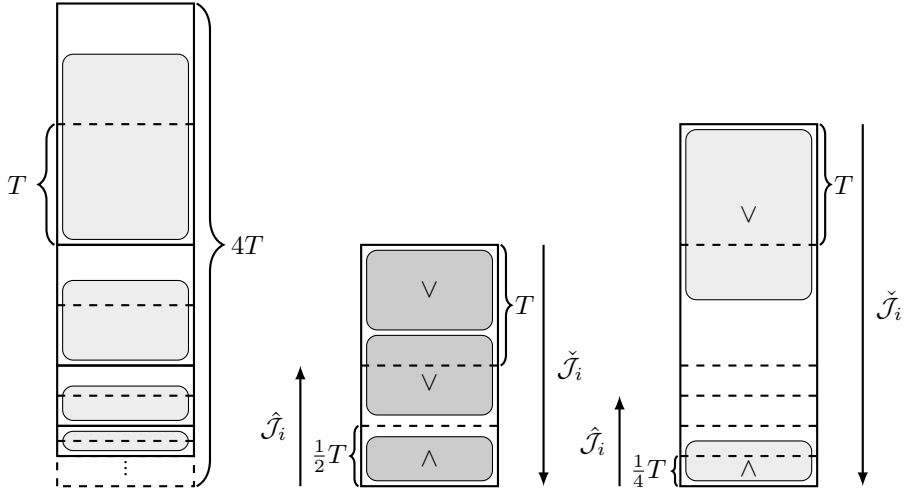

\subparagraph{Ideas and Notation.}

The basic idea of our approach is to use migration to avoid the accumulation of load from previous phases.
Instead, for each machine $i$ there is only old load $\oldload_i$ from the previous phase and new load $\newload_i$ from the current phase. 
Correspondingly, the algorithms maintain a partition of the jobs scheduled on a machine $i$ into \emph{old} jobs $\oldjobs_i$ and \emph{new} jobs $\newjobs_i$.
During each phase jobs are moved from $\oldjobs_i$ to $\newjobs_i$ or migrated away from $i$.
This is done in such a way that at the end of each phase we have a proof that the old guess for the maximum load was too small.
The notation was chosen with the following intuition in mind:
The old jobs form a stack from the bottom up, the new jobs a stack from the top down, the stack of new jobs grows during the phase while the stack of old jobs drops, and the migration must ensure that they never meet. 

There are several, important parameters for the algorithm that determine both the competitive ratio and the (non-)amortized migration factor:
First there is the \enquote{doubling} factor~$\xi$. 
When it gets apparent that the old guess of the maximum load was not high enough it is multiplied by $\xi$.
In the classical algorithm, we have $\xi=2$, hence the term \emph{doubling}.

Next, there is $\gamma<1$.
This parameter is used to restrict the (amortized) migration a job $j$ can cause \emph{directly}, meaning that we consider migration potential machine-wise and a job $j$ assigned to machine $i$ contributes a potential of $\gamma p_j$.
Jobs that are migrated using this potential in turn contribute a potential to the machines they are assigned to leading to even more migration.
The overall migration potential in this process is charged to the initial job and can be upper bounded by a geometric series yielding an overall (amortized) migration factor of $\sum_{i=0}^{\infty}\gamma^{i} = 1/(1-\gamma)$.  

There is a priority queue for the jobs, where jobs with higher processing time have a higher priority.
When a new job $j^*$ arrives an insertion procedure is started and $j^*$ is moved to the initially empty queue.
Jobs are taken from the queue and are inserted on some machine from which some jobs may be removed and moved to the queue.
The insertion procedure terminates when the queue is empty.

Following the notation from \cite{BarNoyFN00}, we call a machine $i$ \emph{eligible} for a job $j$, if $p_j/s_i$ is at most as big as the current guess $T$, 
and \emph{saturated}, if the new load $\newload_{i} = \sum_{j\in\newjobs_{i}} p_j/s_i$ is at least as big as $T$.
Furthermore, we refer to a machine $i$ as $\eta$-eligible for a job $j$, if $p_j/s_i\leq \eta T$, i.e., eligibility is the same as $1$-eligibility.
The third and final parameter of the algorithmic framework is $\eta\geq 1$.

\subparagraph{The Framework.}

As mentioned before, we consider and maintain migration potential machine-wise and guarantee that these potentials can be charged to the migration potentials of the newly arriving jobs. 
For a job $j'$ that is on the first position of the priority queue and has to be scheduled, we consider the slowest $\eta$-eligible machine $i'$ that is not yet saturated.
Next, old jobs $j\in\oldjobs_{i'}$ are moved to $\newjobs_{i'}$ if $p_j\geq \eta^{-1}p_{j'}$.
These jobs are chosen such that all slower machines eligible for them are already saturated.
Therefore, they can be utilized for proving that the current guess of the maximum load is too small in case that a job cannot be inserted.
If $i'$ remains not saturated afterwards, $j'$ contributes a migration potential of $\gamma p_{j'}$ to $i'$ which is used -- together with accumulated potential in the amortized case -- to remove old jobs from $i'$ in non-increasing order which are moved to the queue.
In this case $j'$ is assigned to $i'$ and moved to $\newjobs_{i'}$.
Otherwise, we consider the next slowest eligible machine that is not saturated.
These steps are carefully designed such that the load on each machine can be properly bounded at any time.
Lastly, if every machine is saturated or not eligible for $j'$, the guess $T$ was too small.
Therefore, we multiply it by $\xi$ and for each machine $i$ all assigned jobs are considered old, i.e., we move them from $\newjobs_i$ to $\oldjobs_i$. 
The resulting algorithmic framework is summarized and made precise in \cref{algo:insertion}.
\begin{algorithm}[t]
\While{$Q$ not empty}{
	Take the next job $j'$ from $Q$\;
	\While(\tcc*[f]{search for machine to place job}){true}{
		\eIf(\tcc*[f]{there are candidate machines}){$\tilde{\machs}(\eta,j')\neq \emptyset$}{
			Choose a slowest machine $i'\in \tilde{\machs}(\eta,j')$\;
			\lForEach{$j\in\oldjobs_{i'}$ with $p_j\geq \eta^{-1} p_{j'} $}{move $j$ from $\oldjobs_{i'}$ to $\newjobs_{i'}$}
			\lIf(\tcc*[f]{found suitable machine}){$i'$ is not saturated}{break}
		}(\tcc*[f]{no suitable machine, increase guess}){
		$T=\xi T$\;
		\lForEach{$i\in\machs$}{move all jobs from $\newjobs_{i}$ to $\oldjobs_{i}$ [, and set $\potential_i=0$]}
		}
	}
	$\potential = \gamma p_{j'}$ [$+\ \potential_{i'}$]\;
	\ForEach(\tcc*[f]{migrate jobs}){$j\in\oldjobs_{i'}$ ordered non-increasingly by job size}{
		\If{$p_j\leq \potential$}{
			$\potential = \potential - p_{j}$\;
			move $j$ from $\oldjobs_{i'}$ to $Q$\;
			}
	}
	move $j'$ to $\newjobs_{i'}$ [and set $\potential_{i'} = \potential$] \tcc*[r]{assign job}
}
\caption{
The algorithmic framework.
The parts in brackets relate to the amortized case and in this case a potential $\potential_i$ is maintained for each machine $i$. 
For a job $j$, we denote the set of $\eta$-eligible machines for $j$ that are not saturated as $\tilde{\machs}(\eta,j)$. 
A priority queue $Q$ is maintained.
When the very first job $j_1$ arrives we set $T=p_1/s_1$ and place $j_1$ on the first machine.
Afterwards, when a new job $j^*$ arrives it is placed in $Q$ and the following insertion procedure is used. 
}
\label{algo:insertion}
\end{algorithm}

\subparagraph{First Results.}

We collect some basic results regarding the framework.
During each run of the insertion procedure each job enters the queue at most once due to the ordering of the queue.
Furthermore, each job $j'$ that has to be inserted will be inserted after the guess of the objective value was increased at most $\Oh(\log p_{j'})$ many times, yielding:
\begin{remark}
The procedure terminates and has a polynomial running time.
\end{remark}
Next note that the non-amortized version of the algorithm indeed does not use any migration potential accumulated in previous rounds.
Moreover, each job $j$ gives the machine it is scheduled on a migration potential of $\gamma p_j$ (which has to be used directly or not depending on the case).
Jobs with overall size at most $\gamma p_j$ may be migrated using this potential, in turn adding a migration potential of at most $\gamma^2 p_j$ when inserted.
Hence, the overall migration potential added due to $j$ can be bounded by a geometric series:
\begin{remark}\label{rem:migration_factor}
The (non-)amortized migration factor of the algorithms is upper bounded by $\sum_{k=0}^{\infty}\gamma^k=1/(1-\gamma)$. 
\end{remark}
Lastly, we note that the doubling is justified.
\begin{lemma}\label{lem:doubling_correct}
The algorithmic framework maintains $\xi^{-1}T < \Opt(I^*)$.
\end{lemma}
\begin{proof}
First note that the initial guess for the objective value that is made when the first job arrives is tight and hence the statement is true in the beginning.
We will now show that whenever the guess of the objective value is increased, it was indeed too small.

Note that whenever a job is inserted into the set of new jobs $\newjobs_{i'}$, all slower machines eligible for the job are already saturated.
For the jobs that are newly inserted on a machine $i'$ this is clear, because we always consider a slowest $\eta$-eligible machine that is not saturated.
For the jobs that are moved from $\oldjobs_{i'}$ to $\newjobs_{i'}$ it is easy to see as well:
A job $j$ is moved from $\oldjobs_{i'}$ to $\newjobs_{i'}$ when the assignment of a job $j'$ on $i'$ is considered.
In this case we have $p_j\geq \eta^{-1} p_{j'}$ and all slower machines $\eta$-eligible for $j'$ are saturated.
Now, let $i$ be a slower machine eligible for $j$.
This implies $T\geq p_j/s_{i} \geq \eta^{-1} p_{j'}s_{i}$.
Hence, $j'$ is $\eta$-eligible on $i$ and therefore $i$ has to be saturated. 

Whenever the guess of the maximum load is increased, we have a job $j'$ for which each $\eta$-eligible machine is saturated.
The case in which there is no $\eta$-eligible machine for the job is trivial and hence we assume that $j'$ is $\eta$-eligible on machine $1$.
Since this machine is saturated there have to be new jobs filling this machine.
Each of these jobs (as well as $j'$) may also be eligible on slower machines which have to be saturated as well due to the arguments above. 
Hence, let $m'\leq m$ be the minimal machine such that all machines in $\set{1,\dots, m'}$ are saturated and no job placed as a new job on one of these machines is also eligible on a slower machine.
In any assignment with an objective value at most $T$ all the jobs from $\set{j'} \cup \bigcup_{i\in[m']}\newjobs_i$ have to be placed on $\set{1,\dots, m'}$ but have an overall size of more than $\sum_{i\in[m']}Ts_i$.
Hence no such assignment exists.
\end{proof}

\section{Results}

In the following, we discuss different possibilities for setting the parameters $\xi$, $\gamma$, and $\eta$. 
In each case, we mainly have to bound the overall load on each machine and the old load in particular.
More precisely, we want to upper bound the load on each machine by $(1+\eta)T$ yielding a competitive ratio of $(1+\eta)\xi$.
Hence, the competitive ratio that may be realized using the approach taken in this work is lower bounded by $2$ since $\eta \geq 1$ and $ \xi > 1$.

\subsection{First Amortized Approach}

It is not immediately clear why $\eta > 1$ should be considered since larger values of $\eta$ yield a larger objective value as well.
Hence, we first consider the algorithm for $\eta = 1$ and try to upper bound the load on each machine by $(1+\eta) T = 2T$ at all times. 
For this we first look at some examples to guide the choice of the remaining two parameters.

First, consider the case that in the previous phase two jobs with load roughly $\xi^{-1} T$ have been placed on a given machine and in the current phase another job of size $T$ is placed.
If we have $\gamma < \xi^{-1}$, then it may happen that no job can be migrated and to guarantee a load of at most $2T$ we need to have $\xi \geq 2$ which implies a competitive ratio of at least $4$.
Hence, in order to do better, we consider $\gamma \geq \xi^{-1}$.

Next, consider the case that during the previous phase a machine received two jobs both with size roughly $\xi^{-1} T$ and in the current phase we want to place two new jobs of size roughly $\xi^{-1} \gamma^{-1}T$ (which can happen since $\gamma \geq \xi^{-1}$).
Then it may happen that none of the old jobs can be moved to the set of new jobs since $\gamma < 1$.
Furthermore, it is possible that only one of the old jobs can be migrated resulting in an overall load of roughly $2\xi^{-1} \gamma^{-1}T + \xi^{-1} T$.
Upper bounding this value by $2T$ yields $\xi \geq \gamma^{-1} + \frac{1}{2}$.
This situation is depicted in the second part of \cref{fig:load_structure} and turns out to be a critical example:
\begin{lemma}\label{lem:max_load_bounded_am_1}
The algorithmic framework with $\eta = 1$, $\xi = \gamma^{-1} +\frac{1}{2}$ and $\gamma \in (0,1)$ maintains $\load_i\leq (1+\eta)T$ for each machine $i\in\machs$ in the amortized setting.
\end{lemma}
\begin{proof}
We denote the old load at the beginning of the current phase with $\startload_i$, i.e., the load directly after $T$ was increased or initially defined and all assigned jobs are old.
Furthermore, the load that has been shifted from $\oldjobs_i$ to $\newjobs_i$ during the phase is denoted as $\shiftload_i$.
After a new job was assigned to $i$, we have:
\begin{equation}\label{eq:newload_bounded}
\newload_i\leq (1+\eta) T
\end{equation}
Otherwise we would not have assigned the job.
In the first phase we have $\oldload_i = 0$ and jobs enter $\newjobs_i$ only when a job is newly assigned to $i$ yielding $\load_i = \newload_i + \oldload_i = \newload_i\leq (1+\eta) T$ during the first phase. 
Hence, we can inductively assume:
\begin{equation}\label{eq:oldload}
\oldload_i \leq \startload_i \leq (1+\eta)\xi^{-1}T
\end{equation}
If no new job is assigned to $i$ in the current phase, the claim is trivially true.
Hence, assume that a new job $j'$ has just been assigned to $i$ in the current phase.
Now, if $\oldload_i=0$, we again have $\load_i = \newload_i\leq (1+\eta) T$ due to (\ref{eq:newload_bounded}) and therefore we additionally assume $\oldload_i>0$.

Note that the overall migration potential accumulated in the current phase on machine $i$ equals $\gamma(\newload_i - \shiftload_i)$, while $\potential_i$ denotes the migration potential that has not been used.
We have:
\begin{equation*}
\oldload_i = \startload_i - \shiftload_i - (\gamma(\newload_i - \shiftload_i) - \potential_i) \leq \startload_i - (\gamma \newload_i - \potential_i) 
\end{equation*}
Since $\oldload_i>0$ and $p_j/s_i\leq\xi^{-1}\eta T = \xi^{-1} T$ for $j\in\oldjobs_i$ we additionally get 
$\potential_i \leq \min\set{\xi^{-1}T,\oldload_i}$ 
which in turn yields $\gamma \newload_i \leq \startload_i$ together with the prior equation.
Combining these observations with $\xi = \gamma^{-1} + \frac{1}{2}$ and \cref{eq:oldload}, we have:
\begin{align*}
\load_i &= \newload_i + \oldload_i 
\leq \newload_i + \startload_i - (\gamma \newload_i - \potential_i)
\leq (1 - \gamma)\gamma^{-1}\startload_i + \startload_i + \potential_i
= \gamma^{-1}\startload_i + \potential_i\\
& \leq (2\xi^{-1}\gamma^{-1} + \xi^{-1})T
= \frac{2\gamma^{-1} + 1}{\gamma^{-1}+\nicefrac{1}{2}}T 
= 2T
=(1+\eta)T
\end{align*}
\end{proof}
Hence, for the parameter choice stated in the above lemma we can guarantee:
\[\maxload\leq 2 T < 2 \xi \Opt = (2\gamma^{-1} +1)\Opt\]
Note that $(2\gamma^{-1} +1)\rightarrow 3$ for $\gamma\rightarrow 1$.
Now, if we set  $\eps = 2\gamma^{-1} - 2$ we get a competitive ratio of $3+\eps$ with a migration factor of
$\frac{1}{1-\gamma} = \frac{\gamma}{1-\gamma} + 1 = \frac{1}{\gamma^{-1} - 1} + 1 = 2\eps^{-1} + 1 = \Oh(\eps^{-1})$ (see \cref{rem:migration_factor}).
Hence, we can state the first result:
\begin{theorem}\label{thm:result1}
There is an algorithm for \problem with competitive ratio $3+\eps$, amortized migration factor $2/\eps + 1$ and polynomial running time for each $\eps >0$.
\end{theorem}

\subsection{Non-amortized Approach}

The goal of this section is to transfer the previous result to the setting with non-amortized migration.
We start with a brief discussion of the parameter $\eta$. 
In the non-amortized case with $\eta = 1$, it may happen that all of the migration potential of a newly inserted job goes to waste and additionally none of the old jobs can be moved to the new ones.
For instance, consider the case that the set of old jobs for a given machine contains exactly two jobs of the same size placed in the previous round and in the current round three jobs of roughly the same size but slightly bigger arrive such that the corresponding migration potential does not suffice to move the old jobs.
For $2 \leq \xi < 2.5$ this may happen with jobs of size roughly $\xi^{-1}T$ and for $\xi <2$ with jobs of size roughly $0.5T$. 
In both cases, the load on the machine cannot be bounded by $2T$.

To address this problem the parameter $\eta$ was introduced:
If we choose $\eta \geq \gamma^{-1}$, we can guarantee that at least some, namely, at least half, of the migration potential can be used (as long as old jobs remain). 
This is easy to see, since if we insert a job $j'$ on a machine in this setting, all the remaining old jobs have size at most $\eta^{-1} p_j \leq \gamma p_j$ since this is guaranteed by the algorithm.
However, we still cannot guarantee that more that half of the migration potential is used.
This suggests a choice of $\xi \geq 2\gamma^{-1}$ (see also the third part of \cref{fig:load_structure}).
Hence, we show:
\begin{lemma}\label{lem:max_load_bounded_nonam}
The algorithmic framework with $\eta = \gamma^{-1}$, $\xi = 2\gamma^{-1}$ and $\gamma \in (0,1)$ maintains $\load_i\leq (1+\eta)T$ for each machine $i\in\machs$ in the non-amortized setting.
\end{lemma}
\begin{proof}
We follow the same basic approach as in \cref{lem:max_load_bounded_am_1}, i.e., the old load at the beginning of the current phase is denoted as $\startload_i$ and the load that has been shifted from $\oldjobs_i$ to $\newjobs_i$ during the phase as $\shiftload_i$.
Moreover, we have $\newload_i\leq (1+\eta) T$ (\cref{eq:newload_bounded}) after a new job was assigned to machine $i$, we can inductively assume $\oldload_i \leq \startload_i \leq (1+\eta)\xi^{-1}T$ (\cref{eq:oldload}), and consider the critical case that a new job $j'$ has just been assigned to $i$ in the current phase with some old jobs remaining on the machine.

Remember that whenever we start removing jobs from a machine $i$ due to the assignment of a job $j$, all the jobs from $\oldjobs_i$ have size at most $\gamma p_j$, because otherwise they would have been shifted to $\newjobs_i$.
Since we remove the jobs in non-increasing order by size, either all jobs from $\oldjobs_i$ will be removed, or jobs with summed up size at least $0.5p_j$.
In the considered case we have $\oldload_i > 0$ and therefore get:
\begin{equation*}
\oldload_i \leq \startload_i - \shiftload_i - \frac{1}{2}\gamma(\newload_i - \shiftload_i) \leq \startload_i - \frac{1}{2}\gamma\newload_i
\end{equation*}
This together with $\xi = 2\gamma^{-1}$ and $\eta=\gamma^{-1}$ yields:
\begin{align*}
\load_i & = 
\newload_i + \oldload_i 
\leq \newload_i + \startload_i - \frac{1}{2}\gamma\newload_i
= \startload_i + (1-\frac{1}{2}\gamma)\newload_i\\
& \leq ((1+\gamma^{-1})\xi^{-1} + (1-\frac{1}{2}\gamma)(1+\gamma^{-1}))T
= (1+\eta)T
\end{align*}
\end{proof}
Hence, for the parameter choice stated in the above lemma we can guarantee:
\[\maxload\leq (1+\gamma^{-1}) T < (1+\gamma^{-1}) \xi \Opt = 2(1+\gamma^{-1})\gamma^{-1}\Opt\]
Note that $2(1+\gamma^{-1})\gamma^{-1}\rightarrow 4$ for $\gamma\rightarrow 1$.
Let $\eps = 2(1+\gamma^{-1})\gamma^{-1} - 4$ yielding $2\eps^{-1} = \frac{\gamma}{1 + \gamma^{-1} -2\gamma }$ and note that $1 + \gamma^{-1} -2\gamma = 2( 1- \gamma ) + \gamma^{-1}(1-\gamma) \leq 4( 1- \gamma )$ for $\gamma\in[1/2,1)$.
Hence, we have a competitive ratio of $2(1+\gamma^{-1})\gamma^{-1} = 4 + \eps$ with a migration factor of $\frac{1}{1-\gamma} = \frac{4\gamma}{4(1-\gamma)} + 1 \leq \frac{4\gamma}{1 + \gamma^{-1} -2\gamma } + 1 = 8\eps^{-1} + 1 = \Oh(\eps^{-1})$.
Therefore, we can state the second result:
\begin{theorem}\label{thm:result2}
There is an algorithm for \problem with competitive ratio $4+\eps$, non-amortized migration factor $8\eps^{-1} + 1$ and polynomial running time for each $\eps \in (0,8]$.
\end{theorem}

\subsection{Second Amortized Approach}

We carry the idea of using a value of $\eta > 1$ back to the non-amortized case.
In the critical example for the first non-amortized approach, there were two jobs from the previous phase, two jobs in the current phase, all of them rather big, and only one of the old jobs could be migrated.
Now a choice of $\eta \geq \gamma^{-1}$ guarantees that each newly inserted job can migrate at least one job (if old jobs remain).
Hence, we can avoid this example.
However, it is still not immediately clear that a choice of $\eta > 1$ is useful and other problematic cases emerge:
Say there were three rather big jobs in the previous phase and only two can be migrated or there is only one particularly big job in the new phase and only half of its migration potential can be used.
Both of these cases are dealt with in the following:
\begin{lemma}\label{lem:max_load_bounded_am_2}
The algorithmic framework with $\eta = \gamma^{-1}$, $\xi = \gamma^{-1} + \frac{1}{3}$ and $\gamma \in (0,1)$ maintains $\load_i\leq (1+\eta)T$ for each machine $i\in\machs$ in the amortized setting.
\end{lemma}
\begin{proof}
Like before, we denote the old load at the beginning of the current phase with $\startload_i$ and the load that has been shifted from $\oldjobs_i$ to $\newjobs_i$ during the phase as $\shiftload_i$.
Moreover, we can focus on the case that there is some old load remaining on a given machine $i$ and assume \cref{eq:newload_bounded,eq:oldload} using the same arguments as in \cref{lem:max_load_bounded_am_1}.
Note that \cref{eq:oldload} implies the claim if $\newload_i = 0$ and in the following we consider the two cases that there is exactly one or more than one new job.

\proofsubparagraph{Exactly One New Job.}

If $|\newjobs_{i}| = 1$, either exactly one job was shifted from the old to the new jobs or exactly one new job was inserted after some migration and no job was shifted.
Since the first case is trivial, we focus on the second and essentially use the same argument as in the non-amortized case (see \cref{lem:max_load_bounded_nonam}).
Due to the choice of $\eta$, we can guarantee that at least half of the migration potential of the inserted job was used and therefore have:
\begin{equation*}
\oldload_i \leq \startload_i - \frac{1}{2}\gamma\newload_i 
\end{equation*}
Since $|\newjobs_{i}|=1$, we additionally have $\newload_i \leq \eta T$ yielding:
\begin{align*}
\load_i & = 
\newload_i + \oldload_i 
\leq \newload_i + \startload_i - \frac{1}{2}\gamma\newload_i
= \startload_i + (1-\frac{1}{2}\gamma)\newload_i\\
& \leq ((1+\eta)\xi^{-1} + (1-\frac{1}{2}\gamma)\eta)T
= ((\eta + 1)(\eta + \nicefrac{1}{3})^{-1} + \eta - \nicefrac{1}{2})T\\
&< ((\eta + 1)(\nicefrac{2}{3}\cdot\eta + \nicefrac{2}{3})^{-1} + \eta - \nicefrac{1}{2})T
= (1+\eta)T
\end{align*}

\proofsubparagraph{More Than One New Job.}

If more than one job was inserted into $\newjobs_{i}$, we mostly use the same argument as in the first amortized approach (see \cref{lem:max_load_bounded_am_1}).
In particular, if $\pi_i$ is the remaining unused migration potential, we again have:
\begin{equation*}
\oldload_i = \startload_i - \shiftload_i - (\gamma(\newload_i - \shiftload_i) - \potential_i) \leq \startload_i - (\gamma \newload_i - \potential_i) 
\end{equation*}
However, we can use the choice of $\eta$ to get a better bound for $\potential_i$.
To see this, note that whenever a job is considered for insertion on $i$, all old jobs that are too big to be migrated using only the potential from this job are shifted to the set of new jobs.
Hence, we know that since two jobs are included in $\newjobs_{i}$, at least the two biggest old jobs have been removed from the set of old jobs. 
This implies that the biggest remaining old job has a size of at most $\startload_i / 3 \leq (1 + \eta) \xi^{-1} T /3$ yielding $\potential_i \leq (1 + \eta) \xi^{-1} T /3$.  
Hence, $\xi = \nicefrac{1}{3} + \eta$ yields:
\begin{align*}
\load_i &= \newload_i + \oldload_i 
\leq \newload_i + \startload_i - (\gamma \newload_i - \potential_i)
\leq (1 - \gamma)\gamma^{-1}\startload_i + \startload_i + \potential_i
= \gamma^{-1}\startload_i + \potential_i\\
& \leq ((1+\eta)\eta + (1 + \eta)/3)\xi^{-1}T
= ((1+ \eta)(\nicefrac{1}{3} + \eta))\xi^{-1}T
= (1+\eta)T
\end{align*}
\end{proof}
Hence, the stated parameter choice yields:
\[\maxload\leq 2 T < 2 \xi \Opt = (2\gamma^{-1} + \nicefrac{2}{3})\Opt\]
We have $(2\gamma^{-1} + \nicefrac{2}{3})\rightarrow \nicefrac{8}{3}$ for $\gamma\rightarrow 1$.
Now, if we set  $\eps = 2\gamma^{-1} - 2$ we get a competitive ratio of $\nicefrac{8}{3}+\eps$ with a migration factor of
$\frac{1}{1-\gamma} = \frac{\gamma}{1-\gamma} + 1 = \frac{1}{\gamma^{-1} - 1} + 1 = 2\eps^{-1} + 1 = \Oh(\eps^{-1})$ (see \cref{rem:migration_factor}) yielding the last result:
\begin{theorem}\label{thm:result3}
There is an algorithm for \problem with competitive ratio $8/3+\eps$, amortized migration factor $2/\eps + 1$ and polynomial running time for each $\eps >0$..
\end{theorem}

\section{Conclusion}

We conclude this work with a brief discussion of possible future research. 
First, it would be interesting to investigate whether variants of the present approach can beat the lower bound of $\approx 2.564$ (see \cite{DBLP:journals/mst/EbenlendrS15}) for the pure online variant.
However, to breach the barrier of $2$ it seams likely that different algorithmic approaches are needed. 
A good candidate for this could be the LPT rule for list scheduling which was successfully adapted for online machine covering on identical machines with bounded migration \cite{DBLP:journals/talg/GalvezSV20}. 
Of course, a PTAS result with constant migration for \problem seams especially worthwhile to pursue, in particular since this was achieved for the identical machine case \cite{DBLP:journals/mor/SandersSS09} and PTAS results for the offline setting are well-known \cite{DBLP:journals/siamcomp/HochbaumS88,DBLP:journals/siamdm/Jansen10}.
However, it seems highly non-trivial to adapt the techniques used for the identical machine setting to uniformly related machines except for special cases with, e.g., similar or few different machine speeds.
Lastly, it would be interesting to revisit the results by Westbrook \cite{DBLP:journals/jal/Westbrook00} and Andrews et al. \cite{DBLP:journals/algorithmica/AndrewsGZ99} trying to come up with improved competitive ratios for the dynamic case.



\bibliography{bibibib}

\end{document}